\documentclass[a4paper,envcountsect,envcountsame,runhead, fleqn]{llncs}
\usepackage{etex}
\usepackage{xcolor} 
\usepackage{colortbl}
\usepackage{stmaryrd,amssymb,mathtools} 
\usepackage{microtype}%if unwanted, comment out or use option "draft"
\usepackage{verbatim}
\usepackage{paralist}
\usepackage{multirow}
\usepackage[colorinlistoftodos, textwidth=3cm]{todonotes}
\usepackage{latexsym}
\usepackage{amsmath}
\usepackage{paralist}
\usepackage{verbatim}
\usepackage{mathpartir}
\usepackage{soul}
\usepackage{amsfonts,amssymb,mathtools}
\usepackage[all]{xy}
\usepackage{datetime}

\newcommand{\Env}{\mathit{Env}}

\newcommand{\tletr}{{\tt letrec}}

\newcommand{\tin}{{\tt in}}

\newcommand{\wrt}{{w.r.t.\ }}
\newcommand{\eg}{{e.g.\ }}

\newcommand{\ari}{\mathit{ar}}
\newcommand{\env}{\mathit{env}}
\newcommand{\gentzen}[2]{{\displaystyle \frac{#1}{#2}}}
\newcommand{\gentzent}[3]{{\displaystyle \frac{#1}{#2}}\quad #3}
\newcommand{\dotcup}{\ensuremath{\mathaccent\cdot\cup}}
 \newcommand{\matcheq}{{\ \unlhd\ }}
 
\newcommand{\MBA}{\mathbb{A}}
\newcommand{\LRL}{\mathit{LRL}}
\newcommand{\LRLX}{\mathit{LRLX}}
\newcommand{\LRLXE}{\mathit{LRLXE}}

\newcommand{\FA}{\mathit{FA}}
\newcommand{\BA}{\mathit{BA}}
\newcommand{\AT}{\mathit{AT}}
\newcommand{\LA}{\mathit{LA}}
\newcommand{\tops}{\mathit{tops}}
\newcommand{\Fix}{\mathit{Fix}}
\newcommand{\Var}{\mathit{Var}}

\newcommand{\dom}{\mathit{dom}}
\newcommand{\letrecunify}{{\sc {LetrecUnify}}}
\newcommand{\letrecmatch}{{\sc {LetrecMatch}}}
\newcommand{\letrecenvmatch}{{\sc {LetrecEnvMatch}}}

\newcommand{\letrecdagmatch}{{\sc {LetrecDagMatch}}}
\newcommand{\freshdot}{{\#}}
\newcommand{\FX}{{{\cal F}^x }}

\newcommand{\flatten}{\mathit{flat}}
\newcommand{\NP}{\mathit{NP}}
\newcommand{\gen}[1]{\langle #1 \rangle}

\newcommand{\numberLetrec}{\mu_1}
\newcommand{\numberSize}{\mu_2}
\newcommand{\numberVar}{\mu_3}
\newcommand{\numberOccVar}{\mu_4}
\newcommand{\numberEqsNonX}{\mu_5}
\newcommand{\numberEquations}{\mu_6}

\title{Nominal Unification of Higher Order Expressions with Recursive Let%
\thanks{This research has been partially founded by the MINECO/FEDER projects RASO (TIN2015-71799-C2-1-P) and LoCoS (TIN2015-66293-R) and the UdG project MPCUdG2016/055.}}
\titlerunning{Nominal Unification of cyclic let}
\author{ Manfred Schmidt-Schau{\ss}\inst{1} 
   \and Temur Kutsia \inst{2} 
    \and Jordi Levy \inst{3} 
     \and Mateu Villaret\inst{4}}
   \authorrunning{M. Schmidt-Schau{\ss}, T. Kutsia, J. Levy and M. Villaret}
   \institute{GU Frankfurt, Germany, %%  Inst. for Computer Science, Germany, 
                 \email{schauss@ki.cs.uni-frankfurt.de}  
                 \and
                 RISC, JKU Linz, Austria,  \email{kutsia@risc.jku.at}
                 \and
                  IIIA - CSIC, Spain, \email{levy@iiia.scic.es}
                  \and
                  IMA, Universitat de Girona, Spain, 
                  \email{villaret@ima.udg.edu}}

\begin{document}
\maketitle
\begin{abstract}
A sound and complete algorithm for nominal unification of higher-order expressions with a recursive let is described,
and shown to run in non-deterministic polynomial time. 
We also explore specializations like nominal letrec-matching for plain expressions and for DAGs and determine the complexity of  
corresponding unification problems.
\end{abstract}
Keywords:  Nominal unification, lambda calculus, higher-order expressions, recursive let, operational semantics  
\section{Introduction}

Unification \cite{baadersnyder:2001}  is an operation to make two logical expressions equal by finding substitutions into variables. 
There  are numerous applications in computer science, in particular of (efficient) first-order unification, for example in automated reasoning, 
type checking and verification. 
Unification algorithms are also extended to higher-order calculi with various equivalence relations.  
If  equality includes $\alpha$-conversion and $\beta$-reduction and perhaps also $\eta$-conversion 
of a (typed or untyped) lambda-calculus,  
then unification procedures are known  (see e.g. \cite{huet:75}), however,
 the problem is undecidable \cite{goldfarb:81,levy-veanes:00}.

Our motivation comes from syntactical reasoning on higher-order expressions, with equality being alpha-equivalence of expressions, 
and where a unification algorithm is demanded as a basic service.    
Nominal unification is the extension of first-order unification with abstractions. It unifies expressions  
 \wrt alpha-equivalence, and employs permutations
 as a clean treatment of renamings. 
It is known that  nominal unification is decidable 
in exponential  time \cite{urban-pitts-gabbay:03,DBLP:journals/tcs/UrbanPG04}, where the complexity of the decision problem is   
polynomial time \cite{DBLP:journals/tcs/CalvesF08}. 
It can be seen also from a higher-order perspective \cite{Cheney05,levy-villaret:12}, as equivalent to 
Miller's higher-order pattern unification \cite{DBLP:journals/logcom/Miller91}. 
There are efficient algorithms \cite{DBLP:journals/tcs/CalvesF08,levy-villaret:10}, formalizations of nominal unification
\cite{rincon-fernandez-rocha:16}, formalizations with extensions to commutation properties within expressions 
\cite{ayala-arvalho-fernandez-nantes:16}, and generalizations of nominal unification to narrowing 
  \cite{ayala-fernandez-nantes:16}, and to equivariant (nominal) unification \cite{aoto-kikuchi:16}. 
We are interested in unification \wrt an additional extension with cyclic let.
To the best of our knowledge, there is no nominal unification algorithm for higher-order expressions permitting also general binding 
structures like a cyclic let.

The motivation and intended application scenario is as follows: constructing  syntactic reasoning algorithms for showing properties 
of program transformations on higher-order expressions  
in call-by-need functional languages (see for example \cite{moran-sands-carlsson:99,schmidt-schauss-schuetz-sabel:08})
that have a letrec-construct (also called cyclic let) \cite{ariola-klop-short:94} as in Haskell \cite{haskell2010}, 
(see \eg~\cite{cheney:2005} for a discussion on reasoning with more general name binders, and \cite{urban-kaliszyk:12} for a formalization 
of general binders in Isabelle).  
There may be applications also to coinductive extensions of  logic programming  \cite{simon-mallya-bansal-gupta:06}
 and strict functional languages \cite{jeannin-kozen-silva:12}.  
Basically, 
overlaps of expressions have to be computed (a variant of critical pairs) 
and  reduction steps (under some strategy) have to be performed.
To this end, first an expressive higher-order language is required to represent the meta-notation of expressions. 
For example, the meta-notation  $((\lambda x.e_1)~e_2)$ for a beta-reduction is made operational 
by using unification variables $X_1,X_2$ for $e_1, e_2$.   
The scoping of $X_1$ and $X_2$ is different, which can be dealt with by nominal techniques.
  In fact,
 a more powerful unification algorithm is required for meta-terms employing recursive letrec-environments. 

Our main algorithm \letrecunify\ is derived from  first-order unification and nominal unification: From first-order unification
  we borrowed the decomposition rules,  
and the sharing method from Martelli-Montanari-style unification algorithms \cite{martelli-montanari:82}. The adaptations of decomposition 
for abstractions and the advantageous use of permutations of atoms is derived from nominal unification algorithms.
Decomposing letrec-expression requires an extension by a permutation of the bindings in the environment, where, however, 
one has to take care of scoping. Since in contrast to the basic nominal unification, there are nontrivial fixpoints of permutations
(see Example \ref{example:fixpoints-possible}),  
 novel techniques are required and lead to a surprisingly moderate complexity: 
 a fixed-point shifting rule (FPS) and a redundancy removing rule (ElimFP) together bound the number of 
 fixpoint equations $X \doteq \pi{\cdot}X$ (where $\pi$ is a permutation) using techniques and results from computations in permutation groups. 
  The application of these
  techniques is indispensable (see Example \ref{example:exponential-FPS}) for obtaining efficiency.

 {\em Results}: A nominal letrec unification algorithm \letrecunify\,  
 which is complete and runs in nondeterministic polynomial time  
  (Theorem \ref{thm:unification-terminates}). The nominal letrec unification problem is NP-complete
  (Theorem \ref{thm:matching-NP-hard}).  
  Nominal letrec matching is NP-complete (Theorem \ref{thm:matching-in-NP},\ref{thm:matching-NP-hard}).  
  Nominal letrec matching for dags is in NP and outputs substitutions only (Theorem \ref{thm:nom-dag-matching}), 
  and a very restricted nominal letrec matching problem is graph-isomorphism hard 
  (Theorem \ref{thm:matching-GI-hard}).

\section{The Ground Language of Expressions}   
We define the language $\LRL$ ({\bf L}et{\bf R}ec {\bf L}anguage) of expressions, which is a lambda calculus extended with 
a recursive let construct. The notation is consistent with  \cite{urban-pitts-gabbay:03}. 
The (infinite) set of atoms $\MBA$ is a set of (constant) symbols $a,b$ denoted also with indices (the variables in lambda-calculus). 
There is a set ${\cal F}$ of function symbols with arity $\ari(\cdot)$.
The syntax of the expressions $e$ of $\LRL$ is: \\
\quad $ e   ::=   a \mid \lambda a.e \mid (f~e_1~\ldots e_{\ari(f)}) \mid (\tletr~a_1.e_1; \ldots; a_n.e_n~\tin~e)   
$
 
 We also use tuples, which are written as $(e_1,\ldots,e_n)$, and which are treated as functional expressions 
in the language.~
 We assume that binding atoms $a_1,\ldots,a_n$ in a letrec-expression $(\tletr~a_1.e_1; \ldots; a_n.e_n~\tin~e)$   are pairwise distinct. 
Sequences of bindings $a_1.e_1;\ldots; a_n.e_n~$ are abbreviated as $\env$.

The {\em scope} of atom $a$ in $\lambda a.e$  is standard: $a$ has scope $e$. 
The $\tletr$-construct  has a special scoping rule:
in $(\tletr ~a_1.s_1; \ldots;a_n.s_n~\tin~r)$, every free atom $a_i$ in some $s_j$  or $r$ is bound by the environment 
$a_1.s_1; \ldots;a_n.s_n$. 
This defines the notion of free atoms       
$\FA(e)$, bound atoms $\BA(e)$ in expression $e$, and all atoms $\AT(e)$ in $e$. 
For an environment $\env = \{a_1.e_1,\ldots,a_n.e_n\}$, we define the set of letrec-atoms as $\LA(\env) = \{a_1,\ldots,a_n\}$.  
 We say {\em $a$ is fresh for $e$} iff  $a \not\in \FA(e)$ (also denoted as $a\#e$).
As an example, the expression $(\tletr~f = \mathit{cons}~s_1~g; g = cons~ s_2~ f~\tin ~f)$ represents an infinite list      %defines 
$(cons~s_1~(cons~s_2~(cons~s_1 ~(cons~s_2~\ldots))))$, where $s_1,s_2$ are expressions. However, since our language $\LRL$ is only 
a fragment of 
core calculi  \cite{moran-sands-carlsson:99,schmidt-schauss-schuetz-sabel:08},
the reader may find more programming examples there.

We will use mappings on atoms from $\MBA$. A {\em swapping} $(a~b)$ is a function that maps an atom $a$ to atom $b$, atom $b$ to $a$,   
and is the identity on other atoms.
We will also use finite permutations on atoms from $\MBA$, which are 
  represented as a composition of swappings in the algorithms below. 
Let $\dom(\pi) = \{a \in \MBA \mid \pi(a) \not= a\}$. Then every finite permutation can be represented by a composition of 
at most $(|\dom(\pi)|- 1)$ swappings. 
Composition  $\pi_1 \circ \pi_2$ and inverses $\pi^{-1}$ can be immediately computed. 
Permutations $\pi$ operate on expressions simply by recursing on the structure.   
For a  letrec-expression this is
$\pi\cdot(\tletr ~a_1.s_1; \ldots;a_n.s_n~\tin~e)$ $=$ 
   $(\tletr ~\pi\cdot a_1.\pi\cdot s_1; \ldots;\pi\cdot a_n.\pi\cdot s_n;~\tin~\pi\cdot e)$. 
 Note that permutations also change names of bound atoms.  %swappings and 

We will use the following definition of $\alpha$-equivalence:
\begin{definition}\label{def:alpha-equivalence-sim}
The equivalence $\sim$ on expressions $e \in \LRL$ is defined as follows:
\begin{itemize}
  \item $a \sim a$.  
  \item if $e_i \sim e_i'$ for all $i$, then $f e_1 \ldots e_n \sim f e_1' \ldots e_n'$ for an $n$-ary  $f \in {\cal F}$.
  \item If $e \sim e'$, then $\lambda a.e \sim \lambda a.e'$.
  \item If for $a \not= b$,  $a\#e'$, $e \sim (a~b)\cdot e'$, then $\lambda a.e \sim \lambda b.e'$. 
 \item $\tletr~a_1.e_1; \ldots; a_n.e_n~\tin~e_0 \sim  \tletr~a_1'.e_1'; \ldots; a_n'.e_n'~\tin~e_0'$ iff
   there is some permutation $\rho$ on $\{1,\ldots,n\}$, such that 
      $\lambda a_1.\ldots.\lambda a_n . (e_1,\ldots,e_n,e_0) \sim 
            \lambda a'_{\rho(1)}.\ldots.\lambda a'_{\rho(n)} . (e'_{\rho(1)},\ldots,e'_{\rho(n)},e'_0)$.   \qed
\end{itemize}
\end{definition}
Note that $\sim$ is identical to the equivalence relation generated by $\alpha$-equivalence of binding constructs and permutation of bindings in a 
letrec. 
 
We need fixpoint sets of permutations $\pi$: We define $\Fix(\pi) = \{e\mid \pi\cdot e \sim e\}$.
In usual nominal unification, these sets can be characterized by using freshness constraints  \cite{urban-pitts-gabbay:03}. 
Clearly, all these sets and also all finite intersections are nonempty, 
%%  since at least all atom-free expressions are elements. 
since at least fresh atoms are elements and since $\MBA$ is infinite.
However, in our setting, these sets are nontrivial:
\begin{example}\label{example:fixpoints-possible}   The $\alpha$-equivalence  $(a~b)\cdot (\tletr~c.a;d.b~\tin ~\mathit{True})$ $\sim$ 
$(\tletr~c.a;d.b~\tin ~\mathit{True})$ holds, which means that there
   are expressions $t$ in $\LRL$ with  $t \sim (a~b)\cdot t$ and  $\FA(t) = \{a,b\}$.
\end{example}

  In the following we will use the results on complexity of operations in permutation groups, see
   \cite{luks:91}, and  \cite{Furst-Hopcroft-Luks:80}. We consider a set $\{a_1,\ldots,a_n\}$ of distinct objects
   (in our case the atoms),
   the symmetric group $\Sigma(\{a_1,\ldots,a_n\})$ (of size $n!$) of permutations of the objects, 
   and consider its elements, subsets and subgroups. Subgroups
   are always represented by a set of generators. 
   If $H$ is a set of elements (or generators), then $\gen{H}$ denotes the generated subgroup.  
  Some facts are: 
    \begin{itemize}
      \item Permutations can be represented in space linear in $n$.
      \item Every subgroup of $\Sigma(\{a_1,\ldots,a_n\})$ can be represented by   $\le n^2$ generators.
    \end{itemize}
 However, elements in a subgroup may not be representable as a product of polynomially many generators.
 
 \noindent  The following questions can be answered in polynomial time: 
%   questions 
   \begin{itemize}
     \item The element-question: $\pi \in G$?, 
     \item The subgroup question: $G_1 \subseteq G_2$.  
     \end{itemize} 
   
   However, intersection of groups and set-stabilizer 
    (i.e. $\{\pi \in G \mid \pi(M) = M\}$)
   are not known to be computable in polynomial time, since those problems are as hard as graph-isomorphism (see  \cite{luks:91}).

\section{A Nominal Letrec Unification Algorithm}
%\subsection{The Language for Unification}

As an extension of  $\LRL$, there is also a countably infinite  set of (unification) variables $X,Y$ also denoted perhaps using indices.
The syntax of the language $\LRLX$ ({\bf L}et{\bf R}ec {\bf L}anguage e{\bf X}tended) is
\[\begin{array}{lcl} 
 e & ::=&  a \mid X \mid \pi\cdot{}X \mid \lambda a.e \mid (f~e_1~\ldots e_{\ari(c)})~| ~(\tletr~a_1.e_1; \ldots; a_n.e_n~\tin~e)   
    \end{array}
\] 
$\Var$ is the set of variables and $\Var(e)$ is the set of variables $X$ occurring in $e$.

The expression $\pi{\cdot}e$ for a non-variable $e$ means an operation, which is performed by shifting $\pi$ down, 
using the simplification $\pi_1 {\cdot} (\pi_2 {\cdot} X)$ $\to $ $(\pi_1 \circ  \pi_2) {\cdot}X$, apply it to atoms,
where only  expressions $\pi\cdot{}X$ remain, which are called {\em suspensions}.

A {\em freshness constraint} in our unification algorithm is of the form $a{\freshdot}e$, where $e$ is an $\LRLX$-expression, and an
 {\em atomic} freshness constraint is of the form $a{\freshdot}X$.

 \begin{definition}[{Simplification of Freshness Constraints}]\label{def:simplification-freshness}~ 
 
\begin{flushleft} $\inferrule{\{a{\freshdot}b\} \dotcup \nabla}{\nabla}$ \quad   %~,~ a \not=b
$\inferrule{\{a{\freshdot}(f~s_1 \ldots s_n)\} \dotcup \nabla}{\{a{\freshdot}s_1,\ldots,a{\freshdot}s_n\} \dotcup \nabla}$  \quad 
$\inferrule{\{a{\freshdot}(\lambda a.s)\} \dotcup \nabla}{\nabla}$ \quad 
 $\inferrule{\{a{\freshdot}(\lambda b.s)\} \dotcup \nabla}{\{a{\freshdot}s\} \dotcup \nabla}$~~~  \\[2mm]  %%%   if $a{\not=}b$
 $\inferrule*[right=\text{\normalfont ~\parbox{2.7cm}{if $a \in \{a_1,\ldots,a_n\}$}}]
     {\{a{\freshdot}(\tletr~a_1.s_1;\ldots,a_n.s_n~\tin~r)\} \dotcup \nabla}{\nabla}$  %%  if $a \in \{a_1,\ldots,a_n\}$
 \\[2mm]
 $\inferrule*[right=\text{\normalfont  \parbox{2.7cm}{if $a \not\in \{a_1,\ldots,a_n\}$}}]
   {\{a{\freshdot}(\tletr~a_1.s_1;\ldots,a_n.s_n~\tin~r)\} \dotcup \nabla}
      {\{a{\freshdot}s_1,\ldots a{\freshdot}s_n,a{\freshdot}r\} \dotcup \nabla}$  ~~~ 
   %   if $a{\not\in}\{a_1,\ldots,a_n\}$\qquad 
  $\inferrule*{\{a{\freshdot}(\pi\cdot X)\} \dotcup \nabla}{\{\pi^{-1}(a) {\freshdot} X\} \dotcup \nabla}$  %%  \\ % [0.5mm]
\end{flushleft}
  \end{definition}

\begin{definition}
An $\LRLX$-unification problem is a pair $(\Gamma, \nabla)$, where  $\Gamma$ is a set of equations $s_1 \doteq t_1,\ldots,s_n \doteq t_n$,
 and $\nabla$ is a set of freshness constraints.
A {\em (ground) solution} of $(\Gamma, \nabla)$  is a substitution $\rho$ (mapping variables in $\Var(\Gamma, \nabla)$ to ground expressions), 
 such that $s_i\rho \sim t_i\rho$ for $i = 1,\ldots,n$ 
 %%   (see Definition \ref{def:alpha-equivalence-sim}), 
and for all $a{{\freshdot}}e \in \nabla$:  $a \not\in \FA(e\rho)$ holds. \\
The decision problem is whether there is a solution for given $(\Gamma, \nabla)$.
\end{definition}

\begin{definition}
Let  $(\Gamma, \nabla)$ be an $\LRLX$-unification problem. We consider triples $(\sigma,\nabla',{\cal X})$, where $\sigma$ is 
 a substitution (compressed as a dag) 
mapping variables to $\LRLX$-expressions,
$\nabla'$ is a set of freshness constraints, and  ${\cal X}$ is a set of fixpoint constraints of the form $X \in \Fix(\pi)$,
where $X \not\in \dom(\sigma)$. 
A triple $(\sigma,\nabla',{\cal X})$ is a {\em unifier} of $(\Gamma, \nabla)$, if
  (i) there exists a ground substitution $\rho$ that solves $(\nabla'\sigma,{\cal X})$, i.e.,
for every $a{\#}e$ in $\nabla'$, $a{\#}e\sigma\rho$ is valid, and for every $X \in \Fix(\pi)$ in ${\cal X}$, $X\rho \in \Fix(\pi)$;
and (ii) for every ground substitution $\rho$ that instantiates all variables in $Var(\Gamma, \nabla)$ 
  which solves $(\nabla'\sigma,{\cal X})$,
   the ground substitution $\sigma\rho$ is a solution of  $(\Gamma, \nabla)$. 
   A set $M$ of unifiers is {\em complete}, if every solution $\mu$ is covered by at least one unifier,
   i.e. there is some unifier $(\sigma,\nabla',{\cal X})$ in $M$, and a ground substitution $\rho$, 
   such that $X\mu \sim X\sigma\rho$ for all $X \in \Var(\Gamma, \nabla)$.  
     \qed
\end{definition}

We will employ nondeterministic rule-based algorithms computing unifiers: There is a  clearly indicated subset of disjunctive (don't know non-deterministic) rules.
The {\em collecting variant} of the algorithm  runs  and collects all solutions from all alternatives of the disjunctive rules. 
The {\em decision variant} guesses one possibility and tries to compute 
a single unifier.

 Since we want to avoid the exponential size explosion of the  Robinson-style unification algorithms, keeping  the good properties 
 of Martelli Montanari-style unification algorithms \cite{martelli-montanari:82}, but not their notational overhead,
 we stick to  a set of equations as data structure. 
 As a preparation for the algorithm,  
all expressions in equations are exhaustively flattened as follows:
$(f~t_1 \ldots  t_n) \to (f~X_1 \ldots X_n)$ plus the equations $X_1 \doteq t_1,\ldots,X_n \doteq t_n$. 
Also $\lambda a.s$ is replaced by $\lambda a.X$  with equation $X \doteq s$, 
and $(\tletr ~a_1.s_1;\ldots,a_n.s_n~\tin~r)$ is replaced by  $(\tletr~ a_1.X_1;\ldots,a_n.X_n~\tin~X)$ with the additional equations 
$X_1 \doteq s_1; \ldots ;X_n \doteq s_n;X \doteq r$. The introduced variables are always fresh ones. 
We may denote the resulting set of equations of flattening an equation $\mathit{eq}$ as $\flatten(\mathit{eq})$.
Thus, all expressions in equations are of depth at most 1, where we do not count the permutation 
applications in the suspensions.

 A  dependency  ordering on $\Var(\Gamma)$ is required: 
 If $X  \doteq e$ is in $\Gamma$, and $e$ is not a variable nor a suspension and $X \not= Y \in \Var(e)$, then $X \succ_{vd} Y$,
 Let $>_{vd}$ be the transitive closure of $\succ_{vd}$.
%  We assume that 
 This  ordering is only used, if  no standard rules and no failure rules (see Definition \ref{def:failure-rules}) apply, 
 hence there are no cycles.

\subsection{Rules of the Algorithm \letrecunify}

\letrecunify\  operates on a tuple $(\Gamma, \nabla, \theta)$, where 
$\Gamma$ is a set of flattened equations $e_1 \doteq e_2$, where we assume that  $\doteq$ is symmetric, 
% $\Phi$ is a set of delayed fixpoint equations of the form $X = \pi{\cdot} X$,
$\nabla$ contains freshness constraints,  
$\theta$ represents the already computed substitution as a list of replacements of the form $X \mapsto e$. Initially $\theta$  is empty.
The final state will be reached, i.e. the output,  when $\Gamma$  only contains fixpoint equations of the form $X \doteq \pi{\cdot}X$ that are 
non-redundant, and the rule (Output) fires.  

In the notation of the rules, we use $[e/X]$  as substitution that replaces $X$ by $e$. In the rules, we  may omit $\nabla$ or $\theta$ if they are not changed. 
We will use a  notation ``$|$'' in the consequence part of one rule, perhaps with a set of possibilities, to denote 
 disjunctive  (i.e. don't know) nondeterminism.  The only nondeterministic rule that requires exploring all alternatives is rule (7) below. 
 The other rules can be applied in any order, where it is not necessary to explore alternatives.  %$\Phi$, 

 \paragraph*{Standard (1,2,3,3') and decomposition rules (4,5,6,7):} 

\newcounter{rules}

\begin{flushleft}
 $(1)~\gentzen{\Gamma \dotcup\{e \doteq e\}}{\Gamma}$  \qquad  
  $(2)~\gentzen{\Gamma \dotcup\{\pi\cdot X \doteq s\}~~s \not\in \Var}{\Gamma \dotcup\{ X \doteq \pi^{-1}\cdot s\}}$     %  if $s$ is not a variable. 
  \\[2mm]
  (3)$\gentzen{\Gamma \dotcup\{X \doteq \pi{\cdot}Y\},\nabla, \theta \qquad X \not= Y}{\Gamma[\pi{\cdot}Y/X],  \nabla[\pi{\cdot}Y/X], \theta \cup \{X \mapsto \pi{\cdot}Y\}}$ 
 \quad 
   (3')$\gentzen{\Gamma \dotcup\{X \doteq  Y\},\nabla, \theta \qquad X \not= Y}{\Gamma[Y/X],  \nabla[Y/X], \theta \cup \{X \mapsto Y\}}$ \\[1mm]

$(4)~ \gentzen{\Gamma\dotcup(f~s_1 \ldots s_n) \doteq (f~s_1' \ldots s_n')\}}   %%,\nabla
   {\Gamma \dotcup \{s_1 \doteq s_1',\ldots,s_n \doteq s_n'\}}$\\[2mm]        %% ,\nabla

$(5)~ \gentzen{\Gamma\dotcup(\lambda a.s \doteq \lambda a.t\}}   %,\nabla
      {\Gamma \dotcup \{s \doteq t\}}$    % ,\nabla
\qquad %     \\[2mm]      
     $(6)~ \gentzen{\Gamma\dotcup(\lambda a.s \doteq \lambda b.t\},\nabla}
      {\Gamma \dotcup \{s \doteq (a~b){\cdot}t\} ,\nabla \cup  \{a {\freshdot} t\}}$ 
     \\[2mm]

$(7)~ \gentzen{\Gamma \dotcup\{\tletr ~a_1.s_1;\ldots,a_n.s_n~\tin~r \doteq \tletr ~b_1.t_1;\ldots,b_n.t_n~\tin~r' \}}  % ,\nabla
         {\mathop{|}\limits_{\forall \rho}~\Gamma \dotcup \flatten(\lambda a_1.\ldots \lambda a_n.(s_1,\ldots,s_n,r) \doteq 
          \lambda b_{\rho(1)}.\ldots \lambda b_{\rho(n)}. (t_{\rho(1)},\ldots,t_{\rho(n)},r'))}$\\[2mm]    %% ,~ \nabla
       where $\rho$ is a  permutation on $\{1,\ldots,n\}$. \\
 \end{flushleft}
 \paragraph*{Main Rules:} 
   The following rules (MMS) (Martelli-Montanari-Simulation) and (FPS) (Fixpoint-Shift) will always be immediately followed by a decomposition of 
   the resulting set of equations. \vspace{0.3cm}
 \begin{flushleft}
\mbox{(MMS)} $\gentzen{\Gamma\dotcup \{X \doteq e_1, X \doteq e_2\},\nabla}
   {\Gamma \dotcup  \{X \doteq e_1, e_1 \doteq e_2\}, \nabla}$,~~~
   \begin{minipage}{0.4 \textwidth} if  $e_1, e_2$ are neither variables\\  nor suspensions.
   \end{minipage}
     \\[2mm]
  
  \mbox{(FPS)}~$\gentzen{\Gamma\dotcup \{X \doteq \pi_1{\cdot}X, \ldots,X \doteq \pi_n{\cdot}X, X \doteq e\}, \theta}
   {\Gamma \dotcup  \{e \doteq \pi_1{\cdot}e, \ldots,e \doteq \pi_n{\cdot}e\},  \theta \cup \{X \mapsto e\}}$,   
    \begin{minipage}{0.37 \textwidth}  if $X$ is maximal \wrt $ >_{vd}$, 
     $X \not\in \Var(\Gamma)$,
   and $e$ is neither a variable nor a suspension, and no failure rule (see below) is applicable. 
   \end{minipage}
   \\[2mm]

    \mbox{(ElimFP)}~~$\gentzen{\Gamma\dotcup \{X \doteq \pi_1{\cdot}X, \ldots,X \doteq \pi_n{\cdot}X, X \doteq \pi{\cdot}X\}, \theta}
    {\Gamma\dotcup \{X \doteq \pi_1{\cdot}X, \ldots,X \doteq \pi_n{\cdot}X\}, \theta},$   
    if $\pi \in \gen{\pi_1,\ldots,\pi_n}.$ 
     % \Phi \cup\{X \doteq \pi_1{\cdot}X, \ldots,X \doteq \pi_n{\cdot}X\},   not necessary to remember
   \\[2mm]
   
    \mbox{(Output)}~~$\gentzen{\Gamma, \nabla, \theta}
    { \theta, \nabla, \{``X \in \Fix(\pi)"~|~X \doteq \pi\cdot X \in \Gamma\}}$ \quad  
    \begin{minipage}{0.3\textwidth} if $\Gamma$ only consists \\ of fixpoint-equations.  \end{minipage}
     % \Phi \cup\{X \doteq \pi_1{\cdot}X, \ldots,X \doteq \pi_n{\cdot}X\},   not necessary to remember
 \end{flushleft}
  %% \\[2mm]  
   
   We assume that the rule \mbox{(ElimFP)} will be applied whenever possible. 
   
     Note that the two rules (MMS) and (FPS), without further precaution, may cause an exponential blow-up in the number of fixpoint-equations.
  The rule (ElimFP) will limit the number of fixpoint equations by exploiting knowledge on operations on permutation groups.   
  
  The rule (Output) terminates an execution on $\Gamma_0$ by outputting a unifier $(\theta,\nabla',{\cal X})$.  
  Note that in any case at least one solution is represented:

\vspace*{1em}
The top symbol of an expression is defined as $\tops(X) = X$, $\tops(\pi{\cdot}X) = X$, $\tops(f~s_1 \ldots s_n) = f$,  $\tops(a) = a$,
 $\tops(\lambda a.s) = \lambda$,  
 $\tops(\tletr~\env~\tin~s) = \tletr$. Let $\FX := {\cal F} \cup \MBA \cup\{\tletr, \lambda\}$.

\begin{definition}{\em Failure Rules of \letrecunify} \label{def:failure-rules}
\begin{description}
\item[Clash Failure:] If $s \doteq t \in \Gamma$,  $\tops(s) \in\FX$, $\tops(t) \in \FX$,  but $\tops(s) \not= \tops(t)$.

 \item[Cycle Detection:] If there are equations $X_1 \doteq s_1, \ldots, X_n \doteq s_n$ where $\tops(s_i) \in \FX$,   
 and $X_{i+1}$ occurs in $s_i$ for $i = 1,\ldots,n-1$ and $X_1$ occurs in $s_n$. 
 
 \item[Freshness Fail:] If there is a freshness constraint $a{\freshdot} a$.  % , then FAIL.
 
 \item[Freshness Solution Fail:] If there is a freshness constraint $a{\freshdot} X \in \nabla$, and $a \in \FA((X) \theta)$.  % , then FAIL.
 \end{description}
 \end{definition}
 The computation of  $\FA((X) \theta)$ can be done in polynomial time by iterating over the solution components.

\begin{example}
We illustrate the letrec-rule by a ground example  without flattening. Let the equation be:
 \[\tletr~a.(a,b), b.(a,b)~\tin ~  b \doteq \tletr~b.(b,c), c.(b,c)~\tin ~  c).\]
Select the identity permutation $\rho$, which results in:
\begin{alignat*}{1}
 & \lambda a.\lambda b. ((a,b), (a,b), b) \doteq \lambda b.\lambda c. ((b,c), (b,c),c).\quad \text{Then:}\\
 & \lambda b. ((a,b), (a,b), b) \doteq (a~b){\cdot}\lambda c.((b,c), (b,c),c) =\lambda c.((a,c), (a,c),c).
 \end{alignat*}
(The freshness constraint $a{\freshdot}\ldots$ holds).
Then the application of the  $\lambda$-rule gives $((a,b), (a,b), b) \doteq (b~c){\cdot}((a,c), (a,c),c)$ (the freshness constraint $b{\freshdot}\ldots$ holds).
The resulting equation is
  $((a,b), (a,b), b) \doteq ((a,b), (a,b),b),$ which 
 obviously holds. 
\end{example}

\begin{example}\label{example:exponential-FPS}
This example shows that  FPS (together with the standard and decomposition rules) may give rise to an exponential number of equations on the
size of the original problem.
Let there be variables $X_i, i = 0,\ldots,n$ and the equations 
$\Gamma =  
\{X_n \doteq \pi{\cdot}X_n,$ $X_n \doteq (f~X_{n-1}~\rho_n{\cdot}X_{n-1}),\ldots,  X_2 \doteq  (f~X_{1}~\rho_{2}{\cdot}X_{1})  
\}
$  
 where $\pi,\rho_1,\ldots,\rho_n$ are permutations.

We prove that this unification problem may 
give rise to $2^{n-1}$ many equations, if the redundancy rule (ElimFP) is not there. \\
\begin{tabular}{ll}
The first step is by (FPS): &\hspace*{0.5mm}
\begin{minipage}{0.6\textwidth}
$$ 
\left\{
\begin{array}{rcl}f~X_{n-1}~\rho_n{\cdot}X_{n-1} &\doteq& \pi{\cdot}(f~X_{n-1}~\rho_n{\cdot}X_{n-1}),  \\
    X_{n-1} &\doteq& (f~X_{n-2}~\rho_{n-1}{\cdot}X_{n-2}),  \ldots
    \end{array}
    \right\}
$$ 
\end{minipage}
\end{tabular}\\
 \begin{tabular}{ll}
Using decomposition and inversion:&\hspace*{2mm}
\begin{minipage}{0.5\textwidth}
$$
\left\{
\begin{array}{rcl} X_{n-1}  &\doteq& \pi{\cdot} X_{n-1}, \\
   X_{n-1} &\doteq&\rho_n^{-1}{\cdot}\pi{\cdot} \rho_n{\cdot}X_{n-1}, \\
       X_{n-1} &\doteq &( f~X_{n-2}~\rho_{n-1}{\cdot}X_{n-2}), \ldots
         \end{array}
      \right\}
$$ 
\end{minipage}
\end{tabular}
 \begin{tabular}{ll}
After (FPS)&
\begin{minipage}{0.68\textwidth}
$$
\left\{
\begin{array}{@{}rcl@{}} (f~X_{n-2}~\rho_{n-1}{\cdot}X_{n-2})&\doteq& \pi{\cdot} (f~X_{n-2}~\rho_{n-1}{\cdot}X_{n-2}), \\
     (f~X_{n-2}~\rho_{n-1}{\cdot}X_{n-2})& \doteq&\rho_n^{-1}{\cdot}\pi{\cdot} \rho_n{\cdot}(f~X_{n-2}~\rho_{n-1}{\cdot}X_{n-2}), \\
       X_{n-2} &\doteq &(f~X_{n-3}~\rho_{n-2}{\cdot}X_{n-3}), \ldots\\
         \end{array}
         \right\}
$$ 
\end{minipage}
\end{tabular}\\
 \begin{tabular}{ll}
 decomposition and inversion:&
\begin{minipage}{0.58\textwidth}
$$
\left\{
\begin{array}{rcl}  X_{n-2}   &\doteq& \pi{\cdot}  X_{n-2}, \\
   X_{n-2}  &\doteq& \rho_{n-1}^{-1}{\cdot}\pi{\cdot}  \rho_{n-1}{\cdot}X_{n-2}, \\
     X_{n-2}  &\doteq &\rho_n^{-1}{\cdot}\pi{\cdot} \rho_n{\cdot} X_{n-2}, \\
       X_{n-2} & \doteq &\rho_{n-1}^{-1}{\cdot}\rho_n^{-1}{\cdot}\pi{\cdot} \rho_n{\cdot} \rho_{n-1}{\cdot}X_{n-2}, \\
       X_{n-2} &\doteq& (f~X_{n-3}~\rho_{n-2}{\cdot}X_{n-3}), \ldots\\
         \end{array}
  \right\}
$$ 
\end{minipage}
\end{tabular}
 Now it is easy to see that all equations $X_1 \doteq \pi'{\cdot}X_1$ are generated,  with $\pi' \in \{\rho^{-1}\pi\rho$ where $\rho$ is a composition
 of a subsequence of $\rho_n,\rho_{n-1},\ldots, \rho_2\}$, which makes $2^{n-1}$ equations. 
 The permutations are pairwise different using an appropriate choice of $\rho_i$ and $\pi$. 
 The starting equations can be constructed  using the decomposition rule of abstractions.
\end{example}

   \section{Soundness, Completeness, and Complexity of \letrecunify}

 \begin{theorem}\label{thm:unification-terminates}
 The decision variant of the algorithm \letrecunify\ runs in nondeterministic polynomial time.  
   Its collecting version returns a complete set of at most exponentially many unifiers,
  every one represented in polynomial space.   
 \end{theorem}
 \begin{proof} Note that we assume that the input equations are flattened before applying the rules, which can be performed in polynomial time.
 
  Let $\Gamma_0,\nabla_0$ be the  input, and let $S = \mathit{size}(\Gamma_0,\nabla_0)$.
The execution of a single rule can be done in polynomial time depending on the size of the intermediate state, thus we  have to show that
the size of the intermediate states remains polynomial and that the number of rule applications is at most polynomial. 

 The termination measure $(\numberLetrec,\numberSize,\numberVar,\numberOccVar,\numberEqsNonX,\numberEquations)$, which is ordered lexicographically,  
 is as follows: 
 $\numberLetrec$ is the number of letrec expressions in $\Gamma$,
 $\numberSize$ is the number of letrec-, $\lambda$-symbols, function-symbols and atoms in $\Gamma$, 
 $\numberVar$ is the number of different variables in $\Gamma$,
 $\numberOccVar$ is the number of occurrences of variables in $\Gamma$,
 $\numberEqsNonX$ is the number of equations not of the form $X \doteq e$, and 
 $\numberEquations$ is the number of equations. 
 
Since shifting permutations down and simplification of freshness constraints both terminate and do not increase the measures, we
only compare states which are normal forms for shifting down permutations and simplifying freshness constraints.  
We assume that the algorithm stops if a failure rule is applicable, and 
 that the  rules (MMS) and (FPS) are immediately followed by a full decomposition of the results  (or failure).
 
 Now it is easy to check that the rule applications strictly decrease $\mu$:
 The rules (MMS) and (FPS) together with the subsequent decomposition strictly decrease $(\numberLetrec,\numberSize)$.
 Since expressions in equations are flat, (MMS) does not increase the size: 
  $X \doteq s_1, X \doteq s_2$ is first replaced by  $X \doteq s_1, s_1 \doteq s_2$, and the latter is decomposed, 
  which due to flattening results only in equations containing variables and suspensions. Thus $\numberSize$  is reduced by  the size of $s_2$.
  In the same way (FPS)  strictly decreases $(\numberLetrec,\numberSize)$. In addition $\numberSize$ is at most $S^2$, since only the letrec-decomposition
  rule can add $\lambda a.$-constructs. 
  
   The number of fixpoint-equations for every variable $X$ is at most $c_1*S*log(S))$ for some (fixed) $c_1$, 
   since the  number of atoms is never increased,
   and 
   since we assume that \mbox{(ElimFP)}  is applied whenever possible. 
   The size of the permutation group is at most 
   $S!$, and so the length of proper subset-chains and hence the maximal number of generators of a subgroupp is at most 
   $\log(S!) = O(S*log(S))$. 
     Note that the redundancy of generators  can be tested in polynomial time depending on the number of atoms.

   Now we prove a (global) upper bound  on the number $\numberVar$ of variables: 
    An application of (7)  may increase $\numberVar$ 
   at most by $S$.
   An application of (FPS) may increase this number at  most by $c_1*S \log(S) *S$,
   where the worst case occurs when $e$ is a letrec-expression. 
   Since  (MMS) and (FPS) can be applied at most $S$ times,
   the number of variables is smaller than $c_1*S^3\log(S)$. 
   
   The other rules strictly decrease $(\numberLetrec,\numberSize)$,  or they do not increase $(\numberLetrec,\numberSize)$, 
   and  strictly decrease
   $(\numberVar,\numberOccVar,\numberEqsNonX,\numberEquations)$ and can be performed in polynomial time.   \qed
 \end{proof}

 The problematic rule for complexity is (FPS), which does not increase $\numberLetrec$ and  
 $\numberSize$, but may increase $\numberVar$, $\numberOccVar$ and $\numberEquations$ (see Example \ref{example:exponential-FPS}).
 This increase is defeated by the rule (ElimFP), which helps to keep  the numbers  $\numberOccVar$ and $\numberEquations$ low.

 \begin{theorem}\label{thm:unification-sound-and-complete}
  The algorithm \letrecunify\  is sound and complete.    
  \end{theorem}
\begin{proof}
 Soundness of the algorithm holds, by easy arguments for every rule, similar as in  \cite{urban-pitts-gabbay:03}, and since 
 the letrec-rule follows the definition of $\sim$ in Def. \ref{def:alpha-equivalence-sim}. % is justified in Lemma \ref{lemma:sim-eq-sem-alpha}.
A further argument is that the failure rules are sufficient to detect
 final states without solutions.
 
 Completeness requires more arguments. The decomposition and standard rules (with the exception of rule (7)), 
 retain the set of solutions. The same for (MMS), (FPS),  and \mbox{(ElimFP)}. 
 The nondeterministic Rule (7)  provides all possibilities for potential ground solutions.   
 Moreover, the failure rules are not applicable to states that are solvable.
 
 A final output of  \letrecunify\ has at least one ground solution as instance:   %(which may still contain fixpoint equations) 
 we can instantiate all variables that remain in 
 $\Gamma_{\mathit{out}}$ by a fresh atom. Then all fixpoint equations are satisfied, since the permutations cannot change this atom, 
  and since the (atomic) freshness constraints hold. This ground solution can be represented in polynomial space by using $\theta$, 
  plus an instance $X \mapsto a$ for all remaining  variables $X$ and a fresh atom $a$, and removing all fixpoint equations and 
  freshness constraints. \qed
\end{proof}
 
 \begin{theorem}\label{thm:letrec-unification-in-NP}
 The nominal letrec-unification problem is in $\NP$.
 \end{theorem}
 \begin{proof}
 This follows from Theorems \ref{thm:unification-terminates}  and \ref{thm:unification-sound-and-complete}.
 \end{proof}

  \section{Nominal Matching with Letrec: \letrecmatch}
 
 Reductions in higher order calculi with letrec, in particular on a meta-notation, 
 require a matching algorithm, matching its left hand side  to an expression. 
 
 \begin{example}
 Consider the (lbeta)-rule, which is the version of (beta) used in call-by-need  calculi with sharing
  \cite{ariola:95,moran-sands-carlsson:99,schmidt-schauss-schuetz-sabel:08}.
  \begin{alignat*}{1}
 (\mathit{lbeta}) \quad & (\lambda x.e_1)~e_2 \to ~\tletr~x. e_2~\tin~e_1.
  \end{alignat*}
 An (lbeta) step, for example, on $(\lambda x.x)~(\lambda y.y)$ is performed by switching to the language $LRL$ and then matching
 $(\mathit{app}~(\lambda a.X_1)~X_2) \matcheq  \mathit{app}~ (\lambda a.a)~(\lambda b.b)$, where $\mathit{app}$
 is the explicit representation of the binary application operator. This  results in $\sigma := \{X_1 \mapsto a; X_2 \mapsto (\lambda b.b)\}$, 
 and the reduction result is the $\sigma$-instance of   $(\tletr~a. X_2~\tin~X_1)$, which is  $(\tletr~a. (\lambda b.b)~\tin~a)$.
 Note that only the sharing power of the recursive environment is used here.
 \end{example}
 
 We derive a nominal matching algorithm as a specialization of \letrecunify. 
 We use nonsymmetric  equations written $s \matcheq t$, where $s$ is an $\LRLX$-expression, and $t$ does not contain variables.
 Note that neither freshness constraints nor suspensions are necessary (and hence no fixpoint equations). 
 We assume that the input is a set of equations of (plain) expressions. 
 
  The rules of the algorithm \letrecmatch\  are: \\
  
$\gentzen{\Gamma \dotcup\{e \matcheq e\}}{\Gamma}$ \qquad
$\gentzen{\Gamma\dotcup\{(f~s_1 \ldots s_n)\matcheq (f~s_1' \ldots s_n')\}}
   {\Gamma \dotcup \{s_1 \matcheq s_1',\ldots,s_n \matcheq s_n'\}}$ \qquad
$\gentzen{\Gamma\dotcup\{\lambda a.s \matcheq \lambda a.t\}}
      {\Gamma \dotcup \{s\matcheq t\}}$ 
     \\[2mm] 
     
     $\gentzent{\Gamma\dotcup \{\lambda a.s\matcheq \lambda b.t\}}
      {\Gamma \dotcup \{s \matcheq (a~b){\cdot}t\}}$  
      if  $a \# t$,  otherwise Fail. 
       \\[2mm]

$\gentzen{\Gamma \dotcup\{\tletr ~a_1.s_1;\ldots,a_n.s_n~\tin~r \matcheq \tletr ~b_1.t_1;\ldots,b_n.t_n~\tin~r' \}}
         {\mathop{|}\limits_{\forall \rho}~\Gamma \dotcup\{\lambda a_1.\ldots \lambda a_n.(s_1,\ldots,s_n,r) \matcheq
          \lambda a_{\rho(1)}.\ldots \lambda a_{\rho(n)}.(t_{\rho(1)},\ldots,t_{\rho(n)},r')\}}$\\[2mm]
       where $\rho$ is a (mathematical) permutation on $\{1,\ldots,n\}$\\

%\mbox{(mEq)} 
$\gentzen{\Gamma\dotcup \{X \matcheq e_1, X \matcheq e_2\}}
   {\Gamma \dotcup  \{X \matcheq e_1\}}$
     if $e_1 \sim e_2$, otherwise Fail
   \\[2mm]
   
 The test  $e_1 \sim e_2$ will be performed by the (nondeterministic) matching rules.

\begin{description}
\item[Clash Failure:] if $s \doteq t \in \Gamma$, $\tops(s) \in \FX, \tops(t) \in \FX$, but
  % $s,t$ are not variables nor suspensions, and
 $\tops(s) \not= \tops(t)$.
\end{description}
\noindent

  \begin{theorem}\label{thm:matching-in-NP}
  \letrecmatch\  is sound and complete for nominal letrec matching.
  It decides nominal letrec matching in 
     nondeterministic polynomial time.  
  Its collecting  version  returns a finite complete set of an at most exponential number of matching substitutions, 
  which are of at most polynomial size.   
  \end{theorem}
  \begin{proof}
  This follows by standard arguments.  
  \end{proof}
  
   \begin{theorem}
   Nominal letrec matching is NP-complete. 
  \end{theorem}
  \begin{proof}
  The problem is in NP, which follows from Theorem \ref{thm:matching-in-NP}. It is also NP-hard, 
  which follows from the (independent) Theorem \ref{thm:matching-NP-hard}. 
  \end{proof}
  
   A slightly more general situation for matching occurs, when the matching equations $\Gamma_0$ are compressed using a dag.
   We construct a practically more efficient algorithm \letrecdagmatch\ from \letrecunify\ as follows. First we generate $\Gamma_1$ from $\Gamma_0$, 
   which only contains (plain) flattened  expressions by 
   encoding the dag-nodes as variables together with an equation. An expression is said $\Gamma_0$-ground, if it does not reference variables from $\Gamma_0$ (also via equations).
    In order to avoid suspension (i.e. to have nicer results), the decomposition rule for $\lambda$-expressions with different binder names 
    is modified as follows :

   \begin{flushleft}
   $\gentzen{\Gamma\dotcup(\lambda a.s \doteq \lambda b.t\},\nabla }
      {\Gamma \dotcup \{s \doteq (a~b){\cdot}t\}, \nabla \cup  \{a {\freshdot} t\}}$ 
       ~~ $\lambda b.t$  is $\Gamma_0$-ground  
    %  \qquad {a \not= b} \qquad \lambda b.t \mbox{ does not reference variables from $\Gamma_0$ }
     \\[2mm]  
    \end{flushleft}
 
   The extra conditions $a {\freshdot} t$ and $\Gamma_0$-ground  can be tested in polynomial time.
    The equations $\Gamma_1$ are  processed applying  \letrecunify\ (with the mentioned modification) 
    with the  guidance  that the right-hand sides of match-equations are also right-hand sides of equations in the decomposition rules.
   The resulting matching substitutions can be interpreted as the instantiations into the variables of $\Gamma_0$. 
   Since $\Gamma_0$ is a matching problem, the result will be free of  fixpoint equations, and there will be no freshness constraints in the
   solution.
  Thus we have: 
  \begin{theorem}\label{thm:nom-dag-matching}
  The collecting variant of \letrecdagmatch\ outputs an at most exponential set of dag-compressed substitutions that is complete and where every unifier is
   represented in polynomial space.  
  \end{theorem}

  \section{Hardness of Nominal Letrec Matching and Unification}

  \begin{theorem}\label{thm:matching-NP-hard}
  Nominal letrec matching (hence also unification)  is NP-hard, 
 for two letrec expressions, where subexpressions are free of letrec.
  \end{theorem}
 \begin{proof}
 We encode the NP-hard problem of finding a Hamiltonian cycle in regular graph
 \cite{picouleau:94,garey-johnson-tarjan:76}:
 Let $a_1,\ldots, a_n$ be the vertexes of the graph, and $E$ be the set of edges.
 %  Then $s$ be as follows. 
 The first environment part is $\env_1 = a_1.(node~a_1);\ldots;a_n.(node~a_n)$, and a second environment part 
 $\env_2$ consists of bindings $b. (f~a ~a')$ for every edge $(a,a') \in E$ for fresh names $b$. 
 Then let $s := $
 $(\tletr~\env_1;\env_2~\tin~0)$ representing the graph. \quad
 Let the second expression encode the question whether there is a Hamiltonian cycle in a regular graph as follows.
 The first part of the environment is $\env_1' = a_1.(node~X_1), \ldots, a_n.(node~X_n)$. The second part is 
 $\env_2'$ consisting of $b_1.f~X_1~X_2; b_2.f~X_2~X_3; \ldots b_k.f~X_n~X_1$, and the third part consisting 
 of a  number of (dummy) entries of the form  $b.f~Z_2~Z_3$, where $b$ is always a fresh atom  for every binding, and
  $Z_2, Z_3$  are fresh variables. The number of these dummy entries
 can be easily computed from the number of nodes and the degree of the graph, and it is less than the size of the graph. 
 
 Then the  matching problem is solvable iff the graph has a Hamiltonian cycle.
 \end{proof} 
 
 \begin{theorem}\label{thm:letrec-iunification-NP-complete}
  The nominal letrec-unification problem is NP-complete.
  \end{theorem}
  \begin{proof}
  This follows from Theorems \ref{thm:letrec-unification-in-NP}    and \ref{thm:matching-NP-hard}.
  \end{proof}
  
 We say that an expression $t$ {\em contains garbage}, iff there is a subexpression $(\tletr~\env~\tin ~r)$ ,
 and the environment $\env$ can be split into two environments
 $\env = \env_1;\env_2$, such that $\env_1$ is not trivial, and the atoms from    
 $\LA(\env_1)$ do not occur in $\env_2$ nor in $r$.
 Otherwise, the expression is {\em free of garbage}.  
Since  $\alpha$-equivalence of $\LRL$-expressions is Graph-Isomorphism-complete \cite{schmidtschauss-sabel-rau-RTA:2013}, but  
  $\alpha$-equivalence of garbage-free $\LRL$-expressions is polynomial, it is useful to look for improvements of unification and matching 
for garbage-free expressions.  
As a remark: Graph-Isomorphism is known to have complexity between $\mathit{PTIME}$ and $\NP$; there are arguments that it is weaker than the class of NP-complete problems
   \cite{schoening:88}. There is also a claim that it is quasi-polynomial \cite{Babai:2016}, which means that it  requires less than exponential time.

 \begin{theorem}\label{thm:matching-GI-hard}
  Nominal letrec matching with one occurrence of a single variable and a garbage-free target expression is Graph-Isomorphism-hard.
  \end{theorem}
  \begin{proof}
  Let $G_1, G_2$ be two graphs. Let $s$ be $(\tletr~\env_1~\tin~ f~ a_1 \ldots ~a_n)$ the encoding of  a graph $G_1$ where $\env_1$ is the encoding
  as in the proof of Theorem \ref{thm:matching-NP-hard} and the nodes are 
  encoded as $a_1 \ldots a_n$. Then the expression $s$ is free of garbage. 
  Let the environment $\env_2$  be the encoding of $G_2$ in the expression $t = \tletr~ \env_2~\tin~X$. 
  Then $t$ matches $s$ iff the graphs are isomorphic. Hence we have $GI$-hardness. 
  If there is an isomorphism of $G_1$ and $G_2$, then it is easy to see that this bijection leads to an equivalence of the environments, 
  and we can instantiate $X$ with $(f ~a_1 \ldots a_n)$.   
  \end{proof}

  \section{Nominal Letrec Matching with Environment Variables}
   
   Extending the language by variables $\Env$ that may encode partial letrec-environments would lead to a larger coverage of
   unification problems    %matching and hence reduction and also to 
  in reasoning about the semantics of programming languages.  
  
  \begin{example}
   Consider as an example a rule  (llet-e) that merges letrec environments (see \cite{schmidt-schauss-schuetz-sabel:08}): \\ 
   $(\tletr~\Env_1~\tin~ (\tletr~\Env_2~\tin~ X))~ \to~ (\tletr~\Env_1;\Env_2~\tin~ X))$. \\
   It can be applied to an expression $(\tletr~a.0;b.1~\tin~\tletr~c.(a,b,c)~\tin~c)$ as follows: 
   The left-hand side $(\tletr~\Env_1~\tin~ (\tletr~\Env_2~\tin~ X))$ of the reduction rule   matches 
    $(\tletr~a.0;b.1~\tin~(\tletr~c.(a,b,c)~\tin~c))$ with the match: 
    $\{\Env_1 \mapsto (a.0;b.1); \Env_2 \mapsto c.(a,b,c); X \mapsto c\}$, 
    producing the next expression as an instance of the right hand side $(\tletr~\Env_1;\Env_2~\tin~ X)$, which is: 
    $(\tletr~a.0;b.1;c.(a,b,c) ~\tin~c)$. 
    Note that in the application to extended lambda calculi,
      a bit more care (i.e. a further condition) is needed \wrt scoping in order to get valid reduction results in all cases.
  \end{example}
  
  We will now also have partial environments as syntactic objects.
  
  The grammar for the extended language $\LRLXE$ ({\bf L}et{\bf R}ec {\bf L}anguage e{\bf X}tended with {\bf E}nvironments) is: 
\[\begin{array}{lcl} 
 \env & ::=&  \Env \mid \pi\cdot{}\Env \mid  a.e \mid \env;\env \\ 
 e & ::=&  a \mid X \mid \pi\cdot{}X \mid \lambda a.e \mid (f~e_1~\ldots e_{\ari(c)})~| ~\tletr~\env~\tin~e   
    \end{array}
\]
  
  We define a matching algorithm, where environment variables may
  occur in left hand sides.  This algorithm needs a more expressive
  data structure in equations: a letrec with two
  environment-components, (i) a list of bindings that are already
  fixed in the correspondence to another environment, and (ii) an
  environment that is not yet fixed. We denote the fixed bindings as a
  list, which is the first component.  In the notation we assume that
  the (non-fixed) letrec-environment part on the right hand side may
  be arbitrarily permuted before the rules are applied. The
  justification for this special data structure is the  scoping
  in letrec expressions.
  Note that  suspensions do not occur  in this algorithm.
 
  \begin{definition}\label{def:matchalg-env}
   The matching algorithm  \letrecenvmatch\  for expressions where environment variables $\Env$ and expression variables $X$ 
   may occur only in the left hand sides
  of match equations is described below.
 The rules are:  
 \begin{flushleft}
$\gentzen{\Gamma \dotcup\{e \matcheq e\}}{\Gamma}$ \quad
$\gentzen{\Gamma\dotcup\{(f~s_1 \ldots s_n) \matcheq (f~s_1' \ldots s_n')\}}
   {\Gamma \dotcup \{s_1 \matcheq s_1',\ldots,s_n \matcheq s_n'\}}$ \quad
$\gentzen{\Gamma\dotcup\{\lambda a.s \matcheq \lambda a.t\}}
      {\Gamma \dotcup \{s\matcheq t\}}$ 
      \\[2mm] 
     
     $\gentzent{\Gamma\dotcup \{\lambda a.s\matcheq \lambda b.t\}}
      {\Gamma \dotcup \{s \matcheq (a~b){\cdot}t\}}$   ~if, $a {\freshdot} t$ \mbox{ otherwise Fail}
 \end{flushleft}     
\begin{flushleft}
$\gentzen{\Gamma \dotcup\{(\tletr ~ls;  a.s;\env~\tin~r) \matcheq (\tletr~ ls'; b.t;\env' ~\tin~r') \}}
         {\mathop{|}\limits_{\forall (b.t)}~\Gamma \dotcup\{(\tletr~ ((a.s):ls); \env~\tin~r) \matcheq
           (a~b)(\tletr~ ((b.t):ls';~\env'~\tin~r')\}}$\\[2mm]
           if $a {\freshdot} (\tletr~ ls'; b.t;\env' \tin~r' )$, otherwise Fail.\\[2mm]
 
$\gentzen{\Gamma \dotcup\{(\tletr~ ls;  \Env;\env~\tin~r) \matcheq (\tletr~ ls';  \env_1';\env_2'~\tin~r') \}}
         {\mathop{|}\limits_{\env_1'}~\Gamma \dotcup\{(\tletr~ (\Env:ls); \env~\tin~r)  \matcheq
           (\tletr~ ((env_1'):ls');~\env_2'~\tin~r')\}}$\\[2mm]
           
 $\gentzen{\Gamma \dotcup\left\{
   \begin{array}{l}(\tletr~ ls;\emptyset ~\tin~r) \\ \matcheq
 (\tletr~ ls';\emptyset~ \tin~r')\end{array} \right\}}
         { \Gamma \dotcup\{ls \matcheq ls'; r \matcheq r'\}}$ ~~  
$\gentzen{\Gamma \dotcup\{[e_1;\ldots;e_n] \matcheq [e_1';\ldots;e_n']\}}
         { \Gamma \dotcup\{e_1 \matcheq e_1'; \ldots; e_n \matcheq e_n'\}}$\\[2mm]        
$\gentzen{\Gamma\dotcup \{X \matcheq e_1, X \matcheq e_2\}}
   {\Gamma \dotcup  \{X \matcheq e_1, e_1 \doteq e_2 \}}$   
   \quad  
   $\gentzen{\Gamma\dotcup \{\Env \matcheq \env_1, \Env \matcheq \env_2\}}
   {\Gamma \dotcup  \{\Env \matcheq \env_1, \env_1 \doteq \env_2\}}$ 
 \end{flushleft}       
   
 Testing $e_1 \doteq e_2$ and $\env_1 \doteq \env_2$ is done with high priority using the (nondeterministic) matching rules. 
 \begin{description}
  \item[Clash Failure:] If $s \doteq t \in \Gamma$, $\tops(s)\in \FX, \tops(t) \in \FX$, but
   $\tops(s) \not= \tops(t)$.
\end{description}
 \end{definition}
 
 After successful execution, the  result will be  a set of match equations with components $X\matcheq e$, and $\Env \matcheq \env$,
 which represents a matching substitution.
 
 \begin{theorem}
 The algorithm \ref{def:matchalg-env} (\letrecenvmatch)  is sound and complete. It runs in non-deterministic polynomial time. 
 The corresponding decision problem is NP-complete. 
 The collecting     
  version of \letrecenvmatch\ returns an at most exponentially large, complete set of representations 
 of matching substitutions, where the representations are of at most polynomial size.  
 \end{theorem}
 \begin{proof}
 The reasoning for soundness, completeness and termination in polynomial time is  a variation of  previous arguments. 
 The nonstandard part is fixing the correspondence of environment parts step-by-step and %%  at the same time 
 keeping the scoping.  
 \end{proof}

  \section{Conclusion and Future Research}

 We constructed a nominal letrec unification algorithm, several nominal letrec matching algorithms for variants,
 which all run in nondeterministic polynomial time.
 Future research is to investigate  extensions with environment variables $\Env$, 
 and  to investigate nominal matching together with equational theories.


\begin{thebibliography}{10}
\providecommand{\url}[1]{\texttt{#1}}
\providecommand{\urlprefix}{URL }

\bibitem{aoto-kikuchi:16}
Aoto, T., Kikuchi, K.: A rule-based procedure for equivariant nominal
  unification. In: informal proceedings HOR. p.~5 (2016)

\bibitem{ariola:95}
Ariola, Z.M., Felleisen, M., Maraist, J., Odersky, M., Wadler, P.: A
  call-by-need lambda calculus. In: {POPL}'95. pp. 233--246. ACM Press, San
  Francisco, CA (1995)

\bibitem{ariola-klop-short:94}
Ariola, Z.M., Klop, J.W.: Cyclic {L}ambda {G}raph {R}ewriting. In: Proc. IEEE
  {LICS}. pp. 416--425. IEEE Press (1994)

\bibitem{ayala-arvalho-fernandez-nantes:16}
Ayala-Rinc\'on, M., de~Carvalho-Segundo, W., Fern\'andez, M., Nantes-Sobrinho,
  D.: A formalisation of nominal alpha-equivalence with a and ac function
  symbols. In: Proc. LSFA 2016. pp. 78--93 (2016)

\bibitem{ayala-fernandez-nantes:16}
Ayala-Rinc\'on, M., Fern\'andez, M., Nantes-Sobrinho, D.: Nominal narrowing.
  In: Pientka, B., Kesner, D. (eds.) Proc. first FSCD. pp. 11:1--11:17. LIPIcs
  (2016)

\bibitem{rincon-fernandez-rocha:16}
Ayala-Rinc\'on, M., Fern\'andez, M., Rocha-Oliveira., A.C.: Completeness in pvs
  of a nominal unification algorithm. ENTCS  323(3) (2016), to appear

\bibitem{baadersnyder:2001}
Baader, F., Snyder, W.: Unification theory. In: Robinson, J.A., Voronkov, A.
  (eds.) Handbook of {A}utomated {R}easoning, pp. 445--532. Elsevier and MIT
  Press (2001)

\bibitem{Babai:2016}
Babai, L.: Graph isomorphism in quasipolynomial time. Available from
  http://arxiv.org/abs/1512.03547v2 (2016)

\bibitem{DBLP:journals/tcs/CalvesF08}
Calv{\`e}s, C., Fern{\'a}ndez, M.: A polynomial nominal unification algorithm.
  Theor. Comput. Sci.  403(2-3),  285--306 (2008)

\bibitem{Cheney05}
Cheney, J.: Relating higher-order pattern unification and nominal unification.
  In: Proc. 19th International Workshop on Unification, {UNIF'05}. pp. 104--119
  (2005)

\bibitem{cheney:2005}
Cheney, J.: Toward a general theory of names: Binding and scope. In: MERLIN
  2005. pp. 33--40. ACM (2005)

\bibitem{Furst-Hopcroft-Luks:80}
Furst, M.L., Hopcroft, J.E., Luks, E.M.: Polynomial-time algorithms for
  permutation groups. In: 21st FoCS. pp. 36--41. {IEEE} Computer Society (1980)

\bibitem{garey-johnson-tarjan:76}
Garey, M.R., Johnson, D.S., Tarjan, R.E.: The planar {H}amiltonian circuit
  problem is {NP}-complete. {SIAM} J. Comput.  5(4),  704--714 (1976)

\bibitem{goldfarb:81}
Goldfarb, W.D.: The undecidability of the second-order unification problem.
  Theoretical Computer Science  13,  225--230 (1981)

\bibitem{huet:75}
Huet, G.P.: A unification algorithm for typed lambda-calculus. Theor. Comput.
  Sci.  1(1),  27--57 (1975)

\bibitem{jeannin-kozen-silva:12}
Jeannin, J.B., Kozen, D., Silva, A.: {CoCaml}: Programming with coinductive
  types. Tech. Rep.~{\url{http://hdl.handle.net/1813/30798}}, Computing and
  Information Science, Cornell University (December 2012), fundamenta
  Informaticae, to appear

\bibitem{levy-veanes:00}
Levy, J., Veanes, M.: On the undecidability of second-order unification. Inf.
  Comput.  159(1-2),  125--150 (2000)

\bibitem{levy-villaret:10}
Levy, J., Villaret, M.: An efficient nominal unification algorithm. In: Lynch,
  C. (ed.) Proc. 21st RTA. LIPIcs, vol.~6, pp. 209--226. Schloss Dagstuhl
  (2010)

\bibitem{levy-villaret:12}
Levy, J., Villaret, M.: Nominal unification from a higher-order perspective.
  {ACM} Trans. Comput. Log.  13(2), ~10 (2012)

\bibitem{luks:91}
Luks, E.M.: Permutation groups and polynomial-time computation. In:
  Finkelstein, L., Kantor, W.M. (eds.) Groups And Computation, Proceedings of a
  {DIMACS} Workshop. {DIMACS}, vol.~11, pp. 139--176. {DIMACS/AMS} (1991)

\bibitem{haskell2010}
Marlow, S. (ed.): Haskell 2010 -- Language Report (2010)

\bibitem{martelli-montanari:82}
Martelli, A., Montanari, U.: An efficient unification algorithm. ACM
  Transactions on Programming Languages and Systems  4(2),  258--282 (1982)

\bibitem{DBLP:journals/logcom/Miller91}
Miller, D.: A logic programming language with lambda-abstraction, function
  variables, and simple unification. J. Log. Comput.  1(4),  497--536 (1991)

\bibitem{moran-sands-carlsson:99}
Moran, A.K.D., Sands, D., Carlsson, M.: Erratic fudgets: A semantic theory for
  an embedded coordination language. In: Coordination '99. LNCS, vol. 1594, pp.
  85--102. Springer-Verlag (1999)

\bibitem{picouleau:94}
Picouleau, C.: Complexity of the {H}amiltonian cycle in regular graph problem.
  Theor. Comput. Sci.  131(2),  463--473 (1994)

\bibitem{schmidtschauss-sabel-rau-RTA:2013}
Schmidt-Schau{\ss}, M., Rau, C., Sabel, D.: {Algorithms for Extended
  Alpha-Equivalence and Complexity}. In: van Raamsdonk, F. (ed.) 24th RTA
  2013). LIPIcs, vol.~21, pp. 255--270. Schloss Dagstuhl (2013)

\bibitem{schmidt-schauss-schuetz-sabel:08}
Schmidt-Schau{\ss}, M., Sch\"utz, M., Sabel, D.: Safety of {N}\"ocker's
  strictness analysis. J. Funct. Programming  18(04),  503--551 (2008)

\bibitem{schoening:88}
Sch{\"o}ning, U.: Graph isomorphism is in the low hierarchy. J. Comput. Syst.
  Sci.  37(3),  312--323 (1988)

\bibitem{simon-mallya-bansal-gupta:06}
Simon, L., Mallya, A., Bansal, A., Gupta, G.: Coinductive logic programming.
  In: Etalle, S., Truszczynski, M. (eds.) 22nd ICLP. pp. 330--345. LNCS (2006)

\bibitem{urban-kaliszyk:12}
Urban, C., Kaliszyk, C.: General bindings and alpha-equivalence in nominal
  {I}sabelle. Log. Methods Comput. Sci.  8(2) (2012)

\bibitem{urban-pitts-gabbay:03}
Urban, C., Pitts, A.M., Gabbay, M.: Nominal unification. In: 17th CSL, 12th
  EACSL, and 8th {KGC}. LNCS, vol. 2803, pp. 513--527. Springer (2003)

\bibitem{DBLP:journals/tcs/UrbanPG04}
Urban, C., Pitts, A.M., Gabbay, M.J.: Nominal unification. Theor. Comput. Sci.
  323(1--3),  473--497 (2004)

\end{thebibliography}
\end{document}